\documentclass{llncs}
\usepackage{amssymb,amsmath,subfigure,graphicx, pxfonts}
\usepackage{array}
\usepackage{rotating}
\usepackage{lmodern,blindtext}
\usepackage{epsfig}
\usepackage{multirow}
\usepackage{qtree}
\usepackage{stmaryrd}
\usepackage{ctable}

\begin{document}
\title{An Abductive Framework for Horn Knowledge Base Dynamics}
\author{Radhakrishnan Delhibabu}

\institute{
Informatik 5, Knowledge-Based Systems Group\\
RWTH Aachen, Germany\\
\email{delhibabu@kbsg.rwth-aachen.de}} \maketitle

\begin{abstract}

The dynamics of belief and knowledge is one of the major components
of any autonomous system  that should be able to incorporate new
pieces of information. We introduced the Horn knowledge base
dynamics to deal with two important points: first, to handle belief
states that need not be deductively closed; and the second point is
the ability to declare certain parts of the belief as immutable. In
this paper, we address another, radically new approach to this
problem. This approach is very close to the Hansson's dyadic
representation of belief. Here, we consider the immutable part as
defining a new logical system. By a logical system, we mean that it
defines its own consequence relation and closure operator. Based on
this, we provide an abductive framework for Horn knowledge base
dynamics.

\vspace{0.5cm}

\end{abstract}
\section{Introduction}
Over the last three decades \cite{Den}, abduction has been embraced
in AI as a non-monotonic reasoning paradigm to address some of the
limitations of deductive reasoning in classical logic. The role of
abduction has been demonstrated in a variety of applications. It has
been proposed as a reasoning paradigm in AI for diagnosis, natural
language understanding, default reasoning, planning, knowledge
assimilation and belief revision, multi-agent systems and other
problems (see \cite{Sch}).

In the concept of knowledge assimilation and belief revision (see
\cite{Pag}), when a new item of information is added to a knowledge
base, inconsistency can result. Revision means modifying the Horn
knowledge base in order to maintain consistency, while keeping the
new information and removing (contraction) or not removing the least
possible previous information. In our case, update means revision
and contraction, that is insertion and deletion in database
perspective. Our previous work \cite{Arav1,Arav} makes connections
with contraction from Horn knowledge base dynamics.

Our Horn knowledge base dynamics is defined in two parts: an
immutable part (formulae or sentences) and updatable part (literals)
(for definition and properties see works of Nebel \cite{Nebel} and
Segerberg \cite{Seg}). Horn knowledge bases have a set of integrity
constraints (see the definitions in later section). In the case of
finite Horn knowledge bases, it is sometimes hard to see how the
update relations should be modified to accomplish certain knowledge
base updates.

\begin{example} \label{E1} Consider a database with an (immutable)
rule that a staff member is a person who is currently working in a
research group under a chair. Additional (updatable) facts are that
matthias and gerhard are group chairs, and delhibabu and aravindan
are staff members in group info1. Our first integrity constraint
(IC) is that each research group has only one chair ie. $\forall
x,y,z$ (y=z) $\leftarrow$ group\_chair(x,y) $\wedge$
group\_chair(x,z). Second integrity constraint is that a person can
be a chair for only one research group ie. $\forall x,y,z$
(y=z)$\leftarrow$ group\_chair(y,x) $\wedge$ group\_chair(z,x).
\end{example}

\begin{center}
\underline {Immutable part}: staff\_chair(X,Y)$\leftarrow$
staff\_group(X,Z),group\_chair(Z,Y). \vspace{0.5cm}

\underline{Updatable part}: group\_chair(infor1,matthias)$\leftarrow$ \\
\hspace{2.4cm}group\_chair(infor2,gerhard)$\leftarrow$ \\
\hspace{2.6cm}staff\_group(delhibabu,infor1)$\leftarrow$ \\
\hspace{2.6cm}staff\_group(aravindan,infor1)$\leftarrow$ \\
\end{center}
Suppose we want to update this database with the information,
staff\_chair({delhiba\-bu},{aravindan}), that is

\begin{center}
staff\_chair(\underline{delhibabu},\underline{aravindan})$\leftarrow$
staff\_group(\underline{delhibabu},Z) $\bigwedge$
group\_chair(Z,\underline{aravindan})
\end{center}

If we are restricted to definite clauses, there is only one
plausible way to do this: delhibabu and aravindan belong to groups
infor1, this updating means that we need to delete (remove) matthias
from the database and newly add (insert) aravindan to the database
(aravindan got promoted to the chair of the research group infor1
and he was removed from research group infor1). This results in an
update that is too strong. If we allow disjunctive information into
the database, however, we can accomplish the update by minimal
adding wrt consistency
\begin{center}
staff\_group(\underline{delhibabu},infor1) $\lor$
group\_chair(infor1,\underline{aravindan})
\end{center}
and this option appears intuitively to be correct.

When adding new beliefs to the Horn knowledge base, if the new
belief is violating integrity constraints then belief revision needs
to be performed, otherwise, it is simply added. As we will see, in
these cases abduction can be used in order to compute all the
possibilities and it is \emph{not up to user or system} to choose
among them.

When dealing with the revision of a Horn knowledge base (both
insertions and deletions), there are other ways to change a Horn
knowledge base and it has to be performed automatically also.
Considering the information, change is precious and must be
preserved as much as possible. The \emph{principle of minimal
change} \cite{Herz,Schul} can provide a reasonable strategy. On the
other hand, practical implementations have to handle contradictory,
uncertain, or imprecise information, so several problems can arise:
how to define efficient change in the style of AGM \cite{Alch}; what
result has to be chosen \cite{Lak,Lobo,Nayak1}; and finally,
according to a practical point of view, what computational model to
support for Horn knowledge base revision has to be provided?

Since Horn knowledge base change is one of the main problems arising
in knowledge representation, it has been tackled according to
several points of view. In this article, we consider the immutable
part as defining a new logical system. By a logical system, we mean
that it defines its own consequence relation and closure operator.
Based on this, we provide an abductive framework for belief dynamics
(see \cite{Alis1,Bou,Wall}).

The rest of paper is organized as follows: First we start with
preliminaries along with the concept of logical system and
properties of consequences operator. In Section 3, we introduce Horn
knowledge base dynamics with our logical system. In Section 4, we
explore the relationship of Horn knowledge base dynamics with
coherence approach. In Section 5, we present how Horn knowledge base
dynamics can be realized using abductive explanations. In Section 6,
we give brief overview of related works. In Section 7, we make
conclusions with a summary of our contribution as well as a
discussion of future directions of investigation.


\section{Preliminaries}

A first order language consists of an alphabet $\mathcal{A}$ of a
language $\mathcal{L}$. We assume a countable universe of variables
\textmd{Var}, ranged over x,y,z, and a countable universe of
relation (i.e predicate) symbols, ranged over by $\mathcal{A}$. The
following grammar defines \textmd{FOL}, the language of first order
logic with equality and binary relations:

\begin{center}
$\varphi ::=$ $x=x$ $|$ $a(x,x)$ $|$ $\neg\varphi$ $|$ $\bigvee
\phi$ $|$ $\bigwedge\phi$ $|$ $\exists X:\phi$.
\end{center}

Here $\phi\subseteq FOL$ and $X\subseteq Var$ are finite sets of
formulae and variables, respectively.

\begin{definition}[Normal Logic Program \text{[22]}]
By an alphabet $\mathcal{A}$ of a language $\mathcal{L}$ we mean
disjoint sets of constants, predicate symbols, and function symbols,
with at least one constant. In addition, any alphabet is assumed to
contain a countably infinite set of distinguished variable symbols.
A term over $\mathcal{A}$ is defined recursively as either a
variable, a constant or an expression of the form $f(t_1,...,t_n)$
where f is a function symbol of $\mathcal{A}$, n its arity, and the
$t_i$ are terms. An atom over  $\mathcal{A}$ is an expression of the
form $P(t_1,...,t_n)$ where P is a predicate symbol of $\mathcal{A}$
and the $t_i$ are terms. A literal is either an atom A or its
default negation not A. We dub default literals those of the form
not A. A term (resp. atom, literal) is said ground if it does not
contain variables. The set of all ground terms (resp. atoms) of
$\mathcal{A}$ is called the Herbrand universe (resp. base) of
$\mathcal{A}$. A Normal Logic Program is a possibly infinite set of
rules (with no infinite descending chains of syntactical dependency)
of the form:
\begin{center}
$H\leftarrow B_1,...,B_n,not~C_1,...,not~C_m,~(with~m,n \geq 0 ~and~
finit)$
\end{center}

Where H, the $B_i$ and the $C_j$ are atoms, and each rule stands for
all its ground instances. In conformity with the standard
convention, we write rules of the form $H\leftarrow$ also simply as
H (known as fact). An NLP P is called definite if none of its rules
contain default literals. H is the head of the rule r, denoted by
head(r), and body(r) denotes the set $\{B_1,...,B_n,not~C_1,...,not~
C_m\}$ of all the literals in the body of r.
\end{definition}

When doing problem  modeling with logic programs, rules of the form

\begin{center}
$\bot\leftarrow B_1,...,B_n,not~C_1,...,not~C_m,~(with~m,n \geq 0
~and~ finit)$
\end{center} with a non-empty body are known as a type of integrity constraints
(ICs), specifically denials, and they are normally used to prune out
unwanted candidate solutions. We abuse the $\emph{not}$ default
negation notation applying it to non-empty sets of literals too: we
write not S to denote $\{not ~ s: s\in S\}$, and confound $not~ not~
a\equiv a$. When S is an arbitrary, non-empty set of literals
$S=\{B_1,...,B_n,not~C_1,...,not~ C_m\}$ we use:

\begin{enumerate}
\item[-] $S^+$ denotes the set $\{B_1,\ldots,B_n\}$ of positive literals in $S$ .
\item[-] $S^-$ denotes the set $\{not~ C_1,\ldots, not~ C_m\}$ of negative literals
in $S$ .
\item[-] $|S| = S^+ \cup (not~ S^-)$ denotes the set
$\{B_1,\ldots,B_n,C_1,\ldots ,C_m\}$ of atoms of $S$.
\end{enumerate}

As expected, we say a set of literals $S$ is consistent iff $S^+
\cap |S^-| = \emptyset$. We also write $heads(P)$ to denote the set
of heads of non-IC rules of a (possibly constrained) program $P$,
i.e., $heads(P) = \{head(r) : r \in P\} \backslash \{\bot\}$, and
$facts(P)$ to denote the set of facts of $P$  - $facts(P) =
\{head(r) : r\in P \land body(r) =\emptyset\}$.

\begin{definition}[Level mapping\text{[4]}] Let P be a normal logic program and $B_{P}$ its
Herbrand base. A \emph{level} mapping for P is a function
$\parallel: B_{P} \rightarrow \mathbb{N}$ of ground atoms to natural
numbers. The mapping $\parallel$ is also extended to ground literals
by assigning $\mid\neg A \mid$ = $\mid A\mid$ for all ground atoms
$A\in B_{P}$. For every literal ground L, $\mid L \mid$ is called as
the \emph{level} of L in P.
\end{definition}

\begin{definition}[Acyclic program \text{[4]}] Let P be a normal logic program and $\parallel$ a level mapping
for P.  P is called as acyclic with respect to $\parallel$ if for
every ground clause $H\leftarrow L_{1},...,L_{n}~(with~n \geq 0
~and~ finit)$ in P the level of A is higher then the level of every
$L_{i}$ (1 $\leq $i$ \leq$ n). Moreover P is called acyclic if P is
acyclic with respect to some level mapping for P.
\end{definition}

Unlike Horn knowledge base dynamics, where knowledge is defined as a
set of sentences, here we wish to define a Horn knowledge base KB
wrt a language $\mathcal{L}$, as an abductive framework
$<P,Ab,IC,K>$, where,

\begin{enumerate}
\item[*] $P$ is an acyclic normal logic program with all abducibles in P at level 0
and no non-abducible at level 0. $P$ is referred to as a
\emph{logical system}. This in conjunction with the integrity
constraints corresponds to immutable part of the Horn knowledge
base, here $P$ is defined by immutable part. This is discussed
further in the next subsection;

\item[*] $Ab$ is a set of atoms from $\mathcal{L}$, called the \emph{abducibles}. This notion is required
in an abductive framework, and this corresponds to the atoms that
may appear in the updatable part of the knowledge;

\item[*] $IC$ is the set of \it integrity constraints, \rm a set of sentences from language $\mathcal{L}$.
This specifies the integrity of a Horn knowledge base and forms a
part of the knowledge that can not be modified over time;

\item[*] $K$ is a set of sentences from $\mathcal{L}$. It is the \emph{current knowledge}, and the only part of
$KB$ that changes over time. This corresponds to the updatable part
of the Horn knowledge base. The main requirement here is that no
sentence in $K$ can have an atom that does not appear in $Ab$.
\end{enumerate}

\subsection{Logical system}

The main idea of our approach is to consider the immutable part of
the knowledge to define a new logical system. By a logical system,
we mean that $P$ defines its own consequence relation $\models_{P}$
and its closure $Cn_{p}$. Given $P$, we have the Herbrand Base
$HB_{P}$ and $G_{P}$, the ground instantiation of $P$.

An \it abductive interpretation \rm $I$ is a set of abducibles, i.e.
$I\subseteq Ab$. How $I$ interprets all the ground atoms of $L$
\footnotemark \footnotetext{the set of all the ground atoms of $L$,
in fact depends of $L$, and is given as $HB_{P}$, the Herbrand Base
of P}  is defined, inductively on the level of atoms wrt $P$, as
follows:

\begin{enumerate}
\item[*] An atom $A$ at level 0 (note that only abducibles are at level 0) is interpreted as: $A$ is \it true
\rm in I iff $A\in I$, else it is \it false \rm in $I$.
\item[*] An atom $A$ at level $n$ is interpreted as: $A$ is true in $I$ iff $\exists$ clause $A\leftarrow L_{1},\ldots,L_{k}$
in $G_{P}$ s.t. $\forall L_{j}\;(1\leq j\leq k)$ if $L_{j}$
 is an atom then $L_{j}$ is true in $I$, else if $L_{j}$ is a negative literal
$\neg B_{j}$, then $B_{j}$ is false in I.
\end{enumerate}

This interpretation of ground atoms can be extended, in the usual
way, to interpret  sentences in $L$, as follows (where $\alpha$ and
$\beta$ are sentences):
\begin{enumerate}
\item[*] $\neg \alpha$ is true in $I$ iff $\alpha$ is false in $I$.
\item[*] $\alpha \land  \beta$ is true in $I$ iff both $\alpha$ and $\beta$ are true in $I$.
\item[*] $\alpha \lor  \beta$ is true in $I$ iff either $\alpha$ is true in $I$ or $\beta$ is true in $I$.
\item[*] $\forall \alpha$ is true in $I$ iff all ground instantiations of $\alpha$ are true in $I$.
\item[*] $\exists \alpha$ is true in $I$ iff some ground instantiation of $\alpha$ is true in $I$.
\end{enumerate}

Given a sentence $\alpha$ in $L$, an abductive interpretation $I$ is
said to be an \it abductive model \rm of $\alpha$ iff $\alpha$ is
true in $I$. Extending this to a set of sentences $K$, $I$ is a
abductive model of $K$ iff $I$ is an abductive model of every
sentence $\alpha$ in $K$.

Given a set of sentences $K$ and a sentence $\alpha$, $\alpha$ is
said to be a $P$-\it consequence \rm of $K$, written as $K
\models_{P} \alpha$, iff every abductive model of $K$ is an
abductive model of $\alpha$ also. Putting it in other words, let
$Mod(K)$ be the set of all abductive models of $K$. Then $\alpha$ is
a $P$-consequence of $K$ iff $\alpha$ is true in all abductive
interpretations in $Mod(K)$. The \it consequence operator \rm
$Cn_{P}$ is then defined as $Cn_{P}(K)=\{\alpha~|~K\models _{P}
\alpha\}=\{\alpha~|~\alpha ~ \text{\rm is true in all abductive
interpretations in} ~ Mod(K)\}$. K is said to be P-\it consistent
\rm iff there is no expression $\alpha$ s.t. $\alpha \in Cn_{P}(K)$
and $\neg \alpha\in Cn_{P}(K)$. Two sentences $\alpha$ and $\beta$
are said to be $P$-\it equivalent \rm to each other, written as
$\alpha \equiv \beta$, iff they have the same set of abductive
models , i.e. $Mod(\alpha)=Mod(\beta)$.

\subsection{Properties of consequences operator}

Since a new consequence operator is defined, it is reasonable, to
ask whether it satisfies certain properties that are required in the
Horn knowledge base dynamics context. Here, we observe that all the
required properties, listed by various researchers in Horn knowledge
base dynamics, are satisfied by the defined consequence operator.
The following propositions follow from the above definitions, and
can be verified easily.

\begin{center} $Cn_{P}$ satisfies \it inclusion, i.e. $K\subseteq Cn_{P}(K)$. \end{center}
  \begin{center} $Cn_{P}$ satisfies \it iteration, i.e. $Cn_{P}(K)=Cn_{P}(Cn_{P}(K))$. \end{center}

Anther interesting property is \it monotony, \rm i.e. if $K\subseteq
K'$, then $Cn_{P}(K)\subseteq Cn_{P}(K')$. $Cn_{P}$ satisfies
monotony. To see this, first observe that $Mod(K')\subseteq Mod(K)$.

$Cn_{P}$ satisfies \it superclassicality \rm, i.e. if $\alpha$ can
be derived from K by first order classical logic, then $\alpha \in
Cn_{P}(K)$.

$Cn_{P}$ satisfies \it deduction \rm, i.e. if $\beta \in
Cn_{P}(K\cup \{\alpha\})$, then $(\beta\leftarrow \alpha)\in Cn(K)$.

$Cn_{P}$ satisfies \it compactness \rm, i.e. if $\alpha \in
Cn_{P}(K)$, then $\alpha\in Cn_{P}(K')$ for some finite subset $K'$
of $K$.

\subsection{Statics of a Horn knowledge base}
The statics of a Horn knowledge base $KB$, is given by the current
knowledge K and the integrity constraints $IC$. An abductive
interpretation $M$ is an abductive model of $KB$ iff it is an
abductive model of $K\cup IC$. Let $Mod(KB)$ be the set of all
abductive models of $KB$. The \it belief set \rm represented by
$KB$, written as $KB^{\bullet}$ is given as,
$$KB^{\bullet}=Cn_{P}(K\cup IC)=\{\alpha| \alpha \; \text{\rm is true in
every abductive model of}\; KB \}.$$ A belief (represented by a
sentence in $\mathcal{L}$) $\alpha$ is \emph{accepted} in $KB$ iff
$\alpha\in KB^{|bullet}$ (i.e. $\alpha$ is true in every model of
$KB$). $\alpha$ is \it rejected \rm in $KB$ iff $\neg \alpha \in
KB^{\bullet}$ (i.e. $\alpha$ is false in every model of $KB$). Note
that there may exist a sentence $\alpha$ s.t. $\alpha$ is neither
accepted nor rejected in $KB$ (i.e. $\alpha$ is true in some but not
all models of $KB$), and so $KB$ represents a partial description of
the world.

Two Horn knowledge bases $KB_{1}$ and $KB_{2}$ are said to be \it
equivalent \rm to each other, written as $KB_{1}\equiv KB_{2}$, iff
they are based on the same logical system and their current
knowledge are $P$-equivalent, i.e. $P_{1}=P_{2},\;Ab_{1}=Ab_{2},\;
IC_{1}=IC_{2}$ and $K_{1}\equiv K_{2}$. Obviously, two equivalent
Horn knowledge bases $KB_{1}$ and $KB_{2}$ represent the same belief
set, i.e. $KB^{\bullet}_{1}=KB^{\bullet}_{2}$.


\section{Horn knowledge base dynamics}
In AGM \cite{Alch} three kinds of belief dynamics are defined:
expansion, contraction and revision. We consider all of them, one by
one, in the sequel.

\subsection{Expansion}

Let $\alpha$ be new information that has to be added to a knowledge
base $KB$. Suppose $\neg \alpha$ is not accepted in $KB$. Then,
obviously $\alpha$ is $P$ - consistent with $IC$, and $KB$ can be
\it expanded \rm by $\alpha$, by modifying $K$ as follows:
$$KB+\alpha\equiv <P, Ab, IC, K\cup \{\alpha\}>$$
Note that we do not force the presence of $\alpha$ in the new $K$,
but only say that $\alpha$ must be in the belief set represented by
the expanded Horn knowledge base. If in case $\neg \alpha$ is
accepted in $KB$ (in other words, $\alpha$ is inconsistent with IC),
then expansion of $KB$ by $\alpha$ results in a inconsistent Horn
knowledge base with no abductive models, i.e.
$(KB+\alpha)^{\bullet}$ is the set of all sentences in
$\mathcal{L}$.

Putting it in model-theoretic terms, $KB$ can be expanded by a
sentence $\alpha$, when $\alpha$ is not false in all models of $KB$.
The expansion is defined as:
$$Mod(KB+\alpha)=Mod(KB)\cap Mod(\alpha).$$

If $\alpha$ is false in all models of $KB$, then clearly
$Mod(KB+\alpha)$ is empty, implying that expanded Horn knowledge
base is inconsistent.

\subsection{Revision}

As usual, for revising and contracting a Horn knowledge base, the
rationality of the change is discussed first. Later a construction
is provided that complies with the proposed rationality postulates.

\newpage

\begin{center} \textbf {Rationality postulates} \end{center}

Let $KB=<P,Ab,IC,K>$ be revised by a sentence $\alpha$ to result in
a new Horn knowledge base $KB\dotplus\alpha=<P',Ab',IC',K'>$.

When a Horn knowledge base is revised, we do not (generally) wish to
modify the underlying logical system P or the set of abducibles
$Ab$. This is refereed to as \it inferential constancy \rm by
Hansson \cite{Hans1,Hans2}.

\begin{enumerate}
\item[$(\dotplus1)$] (\it Inferential constancy) $P'=P$ and $Ab'=Ab$,$IC'=IC$.
\item[$(\dotplus 2)$] (\it Success)\rm $\alpha$ is accepted in $KB\dotplus\alpha$ , i.e. $\alpha$ is true in all models of $KB\dotplus\alpha$.
\item[$(\dotplus 3)$] (\it Consistency) \rm $\alpha$ is satisfiable and $P$-consistent with IC iff $KB\dotplus \alpha$ is P-consistent,
i.e. $Mod(\{\alpha\}\cup IC)$ is not empty iff $Mod(KB\dotplus
\alpha)$ is not empty.
\item[$(\dotplus 4)$] (\it Vacuity) \rm If $\neg \alpha$ is not accepted in KB, then $KB\dotplus \alpha\equiv KB+\alpha$, i.e. if $\alpha$
is not false in all models of KB, then $Mod(KB\dotplus
\alpha)=Mod(KB) \cap Mod(\alpha)$.
\item[$(\dotplus 5)$] (\it Preservation)\rm If $KB \equiv KB'$ and $\alpha\equiv \beta$, then $KB\dotplus\alpha\equiv KB'\dotplus\beta$, i.e.
if $Mod(KB)=Mod(KB')$ and $Mod(\alpha)=Mod(\beta)$, then
$Mod(KB\dotplus\alpha)=Mod(KB\dotplus\beta)$.
\item[$(\dotplus 6)$] (\it Extended Vacuity 1)\rm $(KB\dotplus\alpha)+\beta$ implies $KB\dotplus(\alpha \land \beta)$, i.e.
$(Mod(KB\dotplus \alpha)\cap Mod(\beta))\subseteq
Mod(KB\dotplus(\alpha \land \beta))$.
\item[$(\dotplus 7)$] (\it Extended Vacuity 2)\rm If $\neg \beta$ is not accepted in $(KB\dotplus \alpha)$, then
$KB\dotplus(\alpha\land \beta)$ implies $(KB\dotplus \alpha)+\beta$,
i.e. if $\beta$ is not false in all models of $KB\dotplus\alpha$,
then $Mod(KB\dotplus(\alpha\land\beta))\subseteq
(Mod(KB\dotplus\alpha)\cap Mod(\beta))$.
\end{enumerate}

\begin{center}\textbf{Construction}  \end{center}
Let $\mathcal{S}$ stand for the set of all abductive interpretations
that are consistent with $IC$, i.e. $\mathcal{S}=Mod(IC)$. We do not
consider abductive interpretations that are not models of $IC$,
simply because $IC$ does not change during revision. Observe that
when $IC$ is empty, $\mathcal{S}$ is the set of all abductive
interpretations. Given a Horn knowledge base $KB$, and two abductive
interpretations $I_{1}$ and $I_{2}$ from $\mathcal{S}$, we can
compare how close these interpretations are to $KB$ by using an
order $\leq_{KB}$ among abductive interpretations in $\mathcal{S}$.
$I_{1}<_{KB}I_{2}$ iff $I_{1}\leq_{KB}I_{2}$ and $I_{2}\nleq_{KB}
I_{1}$.

Let $\mathcal{F}\subseteq \mathcal{S}$. An abductive interpretation
$I\in \mathcal{F}$ is minimal in $\mathcal{F}$ wrt $\leq_{KB}$ if
there is no $I'\in \mathcal{F}$ s.t. $I'<_{KB}I$. Let,
$Min(\mathcal{F},\leq_{KB})=\{I~|~I~\text{\rm is minimal in}~\\
\mathcal{F}~wrt~\leq_{KB}\}$.

For any Horn knowledge base KB, the following are desired properties
of $\leq_{KB}$:
\begin{enumerate}
\item[($\leq 1$)] (\it Pre-order)\rm $\leq_{KB}$ is a \it pre-order \rm, i.e. it is transitive and reflexive.
\item[($\leq 2$)] (\it Connectivity)\rm $\leq_{KB}$ is \it total \rm in $\mathcal{S}$, i.e. $\forall I_{1},I_{2}\in \mathcal{S}$:
either $I_{1}\leq_{KB} I_{2}$ or $I_{2}\leq_{KB} I_{1}$.
\item[($\leq 3$)] (\it Faithfulness)\rm $\leq_{KB}$ is \it faithful \rm to KB, i.e. $I \in Min(\mathcal{S},\leq_{KB})$ iff
$I \in Mod(KB)$.
\item[($\leq 4$)] (\it Minimality)\rm For any non-empty subset $\mathcal{F}$ of $\mathcal{S}$, $Min(\mathcal{F},\leq_{KB})$
is not empty.
\item[($\leq 5$)] (\it Preservance)] \rm For any Horn knowledge base KB', if $KB\equiv KB'$ then $\leq_{KB}=\leq_{KB'}$.
\end{enumerate}

Let $KB$ (and consequently $K$) be revised by a sentence $\alpha$,
and $\leq_{KB}$ be a rational order that satisfies $(\leq 1)$ to
$(\leq 5)$. Then the abductive models of the revised Horn knowledge
base are given precisely by: $Min(Mod(\{\alpha\}\cup
IC),\leq_{KB})$. Note that, this construction does not say what the
resulting K is, but merely says what should be the abductive models
of the new Horn knowledge base.

\begin{center} \textbf{Representation theorem} \end{center}

Now, we proceed to show that revision of $KB$ by $\alpha$, as
constructed above, satisfies all the rationality postulates
stipulated in the beginning of this section. This is formalized by
the following lemma.

\begin{lemma} Let $KB$ be a Horn knowledge base, $\leq_{KB}$ an order
among $\mathcal{S}$  that satisfies $(\leq 1)$ to $(\leq 5)$. Let a
revision operator $\dotplus$ be defined as: for any sentence
$\alpha$, $Mod(KB\dotplus \alpha)=Min(Mod(\{\alpha\}\cup
IC),\leq_{KB})$. Then $\dotplus$ satisfies all the rationality
postulates for revision $(\dotplus 1)$ to $(\dotplus 7)$.

\begin{proof}~
\begin{enumerate}
\item[$(\dotplus 1)$] $P'=P$ and $Ab'=Ab$ and $IC'=IC$\\ This is
satisfied obviously, since our construction does not touch $P$ and
$Ab$, and $IC$ follows from every abductive interpretation in
$Mod(KB\dotplus \alpha)$.
\item[$(\dotplus 2)$] $\alpha$ is accepted in $KB\dotplus \alpha$\\
Note that every abductive interpretation $M\in Mod(KB+\alpha)$ is a
model of $\alpha$. Hence $\alpha$ is accepted in $KB\dotplus
\alpha$.
\item[$(\dotplus 3)$] $\alpha$ is satisfiable and $P$-consistent with IC iff $KB\dotplus \alpha$ is
$P$-consistent.\\
If part: If $KB\dotplus \alpha$ is $P$-consistent , then
$Mod(KB\dotplus \alpha)$ is not empty. This implies that
$Mod(\{\alpha\}\cup IC)$ is not empty, and hence $\alpha$ is
satisfiable and $P$-consistent with $IC$. \\ Only if part: If
$\alpha$ is satisfiable and P-consistent with $IC$, then
$Mod(\{\alpha\}\cup IC)$ is not empty, and $(\leq 4)$ ensures that
$Mod(KB\dotplus \alpha)$ is not empty. Thus, $KB\dotplus \alpha$ is
$P$-consistent.
\item[$(\dotplus 4)$] If $\neg \alpha$ is not accepted in $KB$, then $KB\dotplus \alpha\equiv
KB+\alpha$.\\
We have to establish that  $Min(Mod(\{\alpha\}\cup
IC),\leq_{KB})=Mod(KB)\cap Mod(\alpha)$. Since $\neg \alpha$ is not
accepted in KB, $Mod(KB)\cap Mod(\alpha)$ is not empty. The required
result follows immediately from the fact that $\leq_{KB}$ is
faithful to KB (i.e. satisfies $\leq 3$), which selects only and all
those models of $\alpha$ which are also models of KB.
\item[$(\dotplus 5)$] If $KB \equiv KB'$ and $\alpha\equiv \beta$ then $KB\dotplus \alpha=KB'\dotplus
\beta$\\
$(\leq 5)$ ensures that $\leq_{KB}=\leq_{KB'}$. The required result
follows immediately from this and the fact that
$Mod(\alpha)=Mod(\beta)$.
\item[$(\dotplus 6)$] $(KB\dotplus \alpha)+\beta$ implies $KB\dotplus (\alpha\land
\beta)$.\\
We consider this in two cases. When $\neg \beta$ is accepted in
$KB\dotplus \alpha$, $(KB\dotplus \alpha)+\beta$ is the set of all
sentences from $\mathcal{L}$, and the postulate follows immediately.
Instead when $\neg \beta$ is not accepted in $KB\dotplus \alpha$,
this postulates coincides with the next one.

\item[$(\dotplus 7)$] If $\neg \beta$ is not accepted in $KB\dotplus \alpha$, then $KB\dotplus (\alpha\land \beta)$ implies  $(KB\dotplus
\alpha)+\beta$.\\
Together with the second case of previous postulate, we need to show
that $KB\dotplus (\alpha\land \beta)=(KB\dotplus \alpha)+\beta$. In
other words, we have to establish that  $Min(Mod(\{\alpha\land
\beta\}\cup IC),\leq_{KB})=Mod(KB\dotplus \alpha)\cap Mod(\beta)$.
For the sake of simplicity, let us represent $Min(Mod(\{\alpha\land
\beta\}\cup IC),\leq_{KB})$ by P, and $Mod(KB\dotplus \alpha)\cap
Mod(\beta)$, which is the same as $Min(Mod(\{\alpha\}\cup
IC),\leq_{KB})\cap Mod(\beta)$, by Q. The required result is
obtained in two parts:
\begin{enumerate}
\item[1)] $\forall$ (abductive interpretation)M: if $M\in P$, then $M\in
Q$\\
Obviously $M\in Mod(\beta)$. Assume that $M\notin
Min(Mod(\{\alpha\}\cup IC),\leq_{KB})$. This can happen in two
cases, and we show that both the cases lead to contradiction. \\
Case A: No model of $\beta$ is selected by $\leq _{KB}$ from
$Mod(\{\alpha\}\cup IC)$. But this contradicts our initial condition
that $\neg \beta$ is not accepted in $KB\dotplus \alpha$.\\
Case B: Some model, say $M'$, of $\beta$ is selected by $\leq_{KB}$
from $Mod(\{\alpha\}\cup IC)$. Since M is not selected, it follows
that $M'<_{KB}M$. But then this contradicts our initial assumption
that $M\in P$. So, $P\subseteq Q$.
\item[2)] $\forall$ (abductive interpretation)M: if $M\in Q$, then $M\in
P$\\
$M \in Q$ implies that $M$ is a model of both $\alpha$ and $\beta$,
and $M$ is selected by $\leq_{KB}$ from $Mod(\{\alpha\}\cup IC)$.
Note that $Mod(\{\alpha\land \beta\}\cup IC)\subseteq
Mod(\{\alpha\}\cup IC)$. Since $M$ is selected by $\leq_{KB}$ in a
bigger set (i.e. $Mod(\{\alpha\}\cup IC)$), $\leq_{KB}$ must select
$M$ from its subset $Mod(\{\alpha\land \beta\}\cup IC)$ also. Hence
$Q\subseteq P$. $\blacksquare$
\end{enumerate}
\end{enumerate}
\end{proof}
\end{lemma}

But, that is not all. Any rational revision of $KB$ by $\alpha$,
that satisfies all the rationality postulates, can be constructed by
our construction method, and this is formalized below.

\begin{lemma} Let $KB$ be a Horn knowledge base and $\dotplus$ a revision
operator that satisfies all the rationality postulates for revision
$(\dotplus 1)$ to $(\dotplus 7)$. Then, there exists an order
$\leq_{KB}$ among $\mathcal{S}$, that satisfies $(\leq 1)$ to $(\leq
5)$, and for any sentence $\alpha$, $Mod(KB\dotplus \alpha)$ is
given in $Min(Mod(\{\alpha\}\cup IC),\leq_{KB})$.\it

\begin{proof} Let us construct an order $\leq_{KB}$ among
interpretations in $\mathcal{S}$ as follows: For any two abductive
interpretations $I$ and $I'$ in $\mathcal{S}$, define $I\leq_{KB}I'$
iff either $I\in Mod(KB)$ or $I \in Mod(KB\dotplus form(I,I'))$,
where $form(I,I')$ stands for sentence whose only models are $I$ and
$I'$. We will show that $\leq_{KB}$ thus constructed satisfies
$(\leq 1)$ to $(\leq 5)$ and $Min(Mod(\{\alpha\}\cup
IC),\leq_{KB})=Mod(KB\dotplus \alpha)$.

First, we show that $Min(Mod(\{\alpha\}\cup
IC),\leq_{KB})=Mod(KB\dotplus \alpha)$.Suppose $\alpha$ is not
satisfiable, i.e. $Mod(\alpha)$ is empty, or $\alpha$ does not
satisfy $IC$, then there are no abductive models of $\{\alpha\}\cup
IC$, and hence $Min(Mod(\{\alpha\}\cup IC),\leq_{KB})$ is empty.
From $(\dotplus 3)$, we infer that $Mod(KB\dotplus \alpha)$ is also
empty. When $\alpha$ is satisfiable and $\alpha$ satisfies $IC$, the
required result is obtained in two parts:

\begin{enumerate}
\item[1)] If $I\in Min(Mod(\{\alpha\}\cup IC),\leq_{KB})$, then $I\in Mod(KB\dotplus
\alpha)$\\
Since $\alpha$ is satisfiable and consistent with $IC$, $(\dotplus
3)$ implies that there exists at least one model, say $I'$, for
$KB\dotplus \alpha$. From $(\dotplus 1)$, it is clear that $I'$ is a
model of $IC$, from $(\dotplus 2)$ we also get that $I'$ is a model
of $\alpha$, and consequently $I\leq_{KB}I'$ (because $I\in
Min(Mod(\{\alpha\}\cup IC),\leq_{KB})$). Suppose $I\in Mod(KB)$,
then $(\dotplus 4)$ immediately gives $I \in Mod(KB\dotplus
\alpha)$. If not, from our definition of $\leq_{KB}$, it is clear
that $I \in Mod(KB\dotplus form(I,I'))$. Note that $\alpha \land
form(I,I')\equiv form(I,I')$, since both $I$ and $I'$ are models of
$\alpha$. From $(\dotplus 6)$ and $(\dotplus 7)$, we get
$Mod(KB\dotplus \alpha)\cap\{I,I'\}=Mod(KB\dotplus form(I,I'))$.
Since $I \in Mod(KB\dotplus form(I,I'))$, it immediately follows
that $I\in Mod(KB\dotplus \alpha)$.
\item[2)] If $I\in Mod(KB\dotplus \alpha)$, then $I \in Min(Mod(\{\alpha\}\cup
IC),\leq_{KB})$.\\
From $(\dotplus 1)$ we get $I$ is a model of $IC$, and from
$(\dotplus 2)$, we obtain $I\in Mod(\alpha)$. Suppose $I\in
Mod(KB)$, then from our definition of $\leq_{KB}$, we get
$I\leq_{KB}I'$, for any other model $I'$ of $\alpha$ and $IC$, and
hence $I \in Min(Mod(\{\alpha\}\cup IC),\leq_{KB})$. Instead, if $I$
is not a model of $KB$, then, to get the required result, we should
show that $I \in Mod(KB\dotplus form(I,I'))$, for every model $I'$
of $\alpha$  and $IC$. As we have observed previously, from
$(\dotplus 6)$ and $(\dotplus 7)$, we get $Mod(KB\dotplus
\alpha)\cap \{I,I'\}=Mod(KB\dotplus form(I,I'))$. Since $I\in
Mod(KB\dotplus \alpha)$, it immediately follows that $I \in
Mod(KB\dotplus form(I,I'))$. Hence $I\leq_{KB} I'$ for any model
$I'$ of $\alpha$ and $IC$, and consequently, $I\in
Min(Mod(\{\alpha\}\cup IC),\leq_{KB})$.
\end{enumerate}

Now we proceed to show that the order $\leq_{KB}$ among
$\mathcal{S}$, constructed as per our definition, satisfies all the
rationality axioms $(\leq 1)$ to $(\leq 5)$.

\begin{enumerate}
\item[$(\leq 1)$] $\leq_{KB}$ is a pre-order.\\
Note that we need to consider only abductive interpretations from
$\mathcal{S}$. From $(\dotplus 2)$ and $(\dotplus 3)$, we have
$Mod(KB\dotplus form(I,I'))=\{I\}$, and so $I \leq_{KB}I$. Thus
$\leq_{KB}$ satisfies reflexivity. let $I_{1}\in Mod(IC)$ and
$I_{2}\notin Mod(IC)$. Clearly, it is possible that two
interpretations $I_{1}$ and $I_{2}$ are not models of $KB$, and
$Mod(KB\dotplus form(I_{1},I_{2})) =\{I_{1}\}$. So, $I_{1}\leq_{KB}
I_{2}$ does not necessarily imply $I_{2}\leq_{KB}I_{1}$, and thus
$\leq_{KB}$ satisfies anti-symmetry.

To show the transitivity, we have to prove that $I_{1}\leq_{KB}
I_{3}$, when $I_{1}\leq_{KB} I_{2}$ and $I_{2}\leq_{KB} I_{3}$ hold.
Suppose $I_{1}\in Mod(KB)$, then  $I_{1}\leq_{KB} I_{3}$ follows
immediately from our definition of $\leq_{KB}$. On the other case,
when $I_{1}\notin Mod(KB)$, we first observe that $I_{1}\in
Mod(KB\dotplus form(I_{1},I_{2}))$, which follows from definition of
$\leq_{KB}$ and $I_{1}\leq_{KB} I_{2}$. Also observe that
$I_{2}\notin Mod(KB)$. If $I_2$ were a model of $KB$, then it
follows from $(\dotplus 4)$ that $Mod(KB\dotplus
form(I_1,I_2))=Mod(KB)\cap \{I_1,I_2\}=\{I_2\}$, which is a
contradiction, and so $I_2\notin Mod(KB)$. This, together with
$I_{2}\leq_{KB}I_{3}$, implies that $I_{2}\in Mod(KB\dotplus
form(I_{2},I_{3}))$. Now consider $Mod(KB+form(I_{1},I_{2},I_{3}))$.
Since $\dotplus $ satisfies $(\dotplus 2)$ and $(\dotplus 3)$, it
follows that this is a non-empty subset of $\{I_{1},I_{2},I_{3}\}$.
We claim that $Mod(KB\dotplus form(I_{1},I_{2},I_{3}))\cap
\{I_{1},I_{2}\}$ can not be empty. If it is empty, then it means
that $Mod(KB\dotplus form(I_{1},I_{2},I_{3}))=\{I_{3}\}$. Since
$\dotplus $ satisfies $(\dotplus 6)$ and $(\dotplus 7)$, this
further implies that $Mod(KB\dotplus
form(I_{2},I_{3}))=Mod(KB\dotplus
form(I_{1},I_{2},I_{3}))\cap\{I_{2},I_{3}\}=\{I_{3}\}$. This
contradicts our observation that $I_{2}\in Mod(KB\dotplus
form(I_{2},I_{3}))$, and so $Mod(KB\dotplus
form(I_{1},I_{2},I_{3}))\cap \{I_{1},I_{2}\}$ can not be empty.
Using $(\dotplus 6)$ and $(\dotplus 7)$ again, we get
$Mod(KB\dotplus form(I_{1},I_{2}))=Mod(KB\dotplus
form(I_{1},I_{2},I_{3}))\cap\{I_{1},I_{2}\}$. Since we know that
$I_{1} \in Mod(KB\dotplus form(I_{1},I_{2}))$, it follows that
$I_{1}\in Mod(KB\dotplus form(I_{1},I_{2},I_{3}))$. From $(\dotplus
6)$ and $(\dotplus 7)$ we also get $Mod(KB\dotplus
form(I_{1},I_{3}))=Mod(KB+form(I_{1},I_{2},I_{3}))\cap
\{I_{1},I_{3}\}$, which clearly implies that $I_{1} \in
Mod(KB\dotplus form(I_{1},I_{3}))$. From our definition of
$\leq_{KB}$, we now obtain $I_{1}\leq_{KB}I_{3}$. Thus,  $\leq_{KB}$
is a pre-order.
\item[$(\leq 2)$] $\leq_{KB}$ is total.\\
Since $\dotplus $ satisfies $(\dotplus 2)$ and $(\dotplus 3)$, for
any two abductive interpretations $I$ and $I'$ in $\mathcal{S}$, it
follows that $Mod(KB\dotplus form(I,I'))$ is a non-empty subset of
$\{I,I'\}$. Hence, $\leq_{KB}$ is total.
\item[$(\leq 3)$] $\leq_{KB}$ is faithful to $KB$.\\
From our definition of $\leq_{KB}$, it follows that $\forall
I_{1},I_{2} \in Mod(KB):I_{1}<_{KB}I_{2}$ does not hold. Suppose
$I_{1}\in Mod(KB)$ and $I_{2}\notin Mod(KB)$. Then, we have
$I_{1}\leq_{KB}I_{2}$. Since $\dotplus $ satisfies $(\dotplus 4)$,
we also have $Mod(KB\dotplus form(I_{1},I_{2}))=\{I_{1}\}$. Thus,
from our definition of $\leq_{KB}$, we can not have
$I_{2}\leq_{KB}I_{1}$. So, if $I_{1}\in Mod(KB)$ and $I_{2}\notin
Mod(KB)$, then $I_{1}<_{KB}I_{2}$ holds. Thus, $\leq_{KB}$ is
faithful to $KB$.
\item[$(\leq 4)$] For any non-empty subset $\mathcal{F}$ of $\mathcal{S}$, $Min(\mathcal{F},\leq_{KB})$
is not empty.\\ Let $\alpha$ be a sentence such that
$Mod(\{\alpha\}\cup IC)=\mathcal{F}$. We have already shown that
$Mod(KB\dotplus \alpha)=Min(\mathcal{F},\leq_{KB})$. Since,
$\dotplus $ satisfies $(\dotplus 3)$, it follows that
$Mod(KB\dotplus \alpha)$ is not empty, and thus
$Min(\mathcal{F},\leq_{KB})$ is not empty.
\item[$(\leq 5)$] If $KB\equiv KB'$, then $\leq_{KB}=\leq_{KB'}$.\\
This follows immediately from the fact that $\dotplus $ satisfies
$(\dotplus 5)$.
\end{enumerate}
\end{proof}
\end{lemma}

Thus, the order among interpretations $\leq_{KB}$, constructed as
per our definition, satisfies $(\leq 1)$ to $(\leq 5)$, and
$Mod(KB\dotplus \alpha)=Min(Mod(\{\alpha\}\cup IC),\leq_{KB}).$
$\blacksquare$

So, we have a one to one correspondence between the axiomatization
and the construction, which is highly desirable, and this is
summarized by the following \it representation theorem. \rm

\begin{theorem} Let $KB$ be revised by $\alpha$, and $KB\dotplus
\alpha$ be obtained by the construction discussed above. Then,
$\dotplus$ is a revision operator iff it satisfies all the
rationality postulates $(\dotplus 1)$ to $(\dotplus 7)$.
\end{theorem}

\begin{proof} Follows from Lemma 1. and Lemma 2. $\blacksquare$
\end{proof}

\subsection{Contraction}
Contraction of a sentence from a Horn knowledge base $KB$ is studied
in the same way as that of revision. We first discuss the
rationality of change during contraction and proceed to provide a
construction for contraction using duality between revision and
contraction.

\begin{center} \textbf{Rationality Postulates} \end{center}

Let $KB=<P,Ab,IC,K>$ be contracted by a sentence $\alpha$ to result
in a new Horn knowledge base $KB\dot{-}\alpha=<P',Ab',IC',K'>$.

\begin{enumerate}
\item[$(\dot{-}1)$] (\it Inferential Constancy)\rm $P'=P$ and $Ab'=Ab$ and $IC'=IC$.
\item[$(\dot{-}2)$] (\it Success)\rm If $\alpha \notin
Cn_{P}(KB)$, then $\alpha$ is not accepted in $KB\dot{-}\alpha$,
i.e. if $\alpha$ is not true in all the abductive interpretations,
then $\alpha$ is not true in all abductive interpretations in
$Mod(KB\dot{-}\alpha)$.
\item[$(\dot{-}3)$](\it Inclusion)\rm $\forall$ (belief)
$\beta$:if $\beta$ is accepted in $KB\dot{-}\alpha$, then $\beta$ is
accepted in $KB$, i.e. $Mod(KB)\subseteq Mod(KB\dot{-}\alpha)$.
\item[$(\dot{-}4)$](\it Vacuity)\rm If $\alpha$ is not
accepted in $KB$, then $KB\dot{-}\alpha=KB$, i.e. if $\alpha$ is not
true in all the abductive models of $KB$, then
$Mod(KB\dot{-}\alpha)=Mod(KB)$.
\item[$(\dot{-}5)$](\it Recovery)\rm $(KB\dot{-}\alpha)+\alpha$
implies $KB$, i.e. $Mod(KB\dot{-}\alpha)\cap Mod(\alpha)\subseteq
Mod(KB)$.
\item[$(\dot{-}6)$](\it Preservation)\rm If $KB\equiv KB'$ and
$\alpha\equiv \beta$, then $KB\dot{-}\alpha=KB'\dot{-}\beta$, i.e.
if $Mod(KB)=Mod(KB')$ and $Mod(\alpha)=Mod(\beta)$, then
$Mod(KB\dot{-}\alpha)=Mod(KB'\dot{-}\beta)$.
\item[$(\dot{-}7)$] (\it Conjunction 1) \rm$KB\dot{-}(\alpha\land \beta)$
implies $KB\dot{-}\alpha\cap KB\dot{-}\beta$, i.e.
$Mod(KB\dot{-}(\alpha\land \beta))\subseteq Mod(KB\dot{-}\alpha)\cup
Mod(KB\dot{-}\beta)$.
\item[$(\dot{-}8)$] (\it Conjunction 2)\rm  If $\alpha$ is not accepted in
$KB\dot{-}(\alpha\land \beta)$, then $KB\dot{-}\alpha$ implies
$KB\dot{-}(\alpha\land \beta)$, i.e. if $\alpha$ is not true in all
the models of $KB\dot{-}(\alpha\land \beta)$, then
$Mod(KB\dot{-}\alpha)\subseteq Mod(KB\dot{-}(\alpha\land \beta))$.
\end{enumerate}

Before providing a construction for contraction, we wish to study
the duality between revision and contraction. The Levi and Harper
identities still holds in our case, and is discussed in the sequel.

\begin{center} \textbf{Relationship between contraction and
revision} \end{center}

Contraction and revision are related to each other. Given a
contraction function $\dot{-}$, a revision function $\dotplus$ can
be obtained as follows:
$$\text{(\it Levi
Identity)}~~~~~Mod(KB\dotplus \alpha)=Mod(KB\dot{-}\neg\alpha)\cap
Mod(\alpha)$$ The following theorem formally states that Levi
identity holds in our approach.

\begin{theorem} Let $\dot{-}$ be a contraction operator that
satisfies all the rationality postulates $(\dot{-}1)$ to
$(\dot{-}8)$. Then, the revision function $\dotplus$, obtained from
$\dot{-}$ using the Levi Identity, satisfies all the rationality
postulates $(\dotplus 1)$ to $(\dotplus 7)$. $\blacksquare$.
\end{theorem}

Similarly, a contraction function $\dot{-}$ can be constructed using
the given revision function $\dotplus $ as follows:
$$\text{(\it Harper
Identity)}~~~~~Mod(KB\dot{-}\alpha)=Mod(KB)\cup Mod(KB\dotplus \neg
\alpha)$$

\begin{theorem} Let $\dotplus $ be a revision operator that
satisfies all the rationality postulates $(\dotplus 1)$ to
$(\dotplus 7)$. Then, the contraction function $\dot{-}$, obtained
from $\dotplus $ using the Harper Identity, satisfies all the
rationality postulates $(\dot{-}1)$ to $(\dot{-}8)$. $\blacksquare$
\end{theorem}

\begin{center}\textbf{Construction} \end{center}

Given the construction for revision, based on order among
interpretation in $\mathcal{S}$, a construction for contraction can
be provided as:

$$Mod(KB\dot{-}\alpha)=Mod(KB)\cup Min(Mod(\{\neg\alpha\}\cup
IC),\leq_{KB}),$$ where $\leq_{KB}$ is the relation among
interpretations in $\mathcal{S}$ that satisfies the rationality
axioms $(\leq 1)$ to $(\leq 5)$. As in the case of revision, this
construction says what should be the models of the resulting Horn
knowledge base, and does not explicitly say what the resulting Horn
knowledge base is.

\newpage

\begin{center}\textbf{Representation theorem} \end{center}

Since the construction for contraction is based on a rational
contraction for revision, the following lemmae and theorem follow
obviously.

\begin{lemma} Let $KB$ be a Horn knowledge base, $\leq_{KB}$ an order
among $\mathcal{S}$ that satisfies $(\leq 1)$ to $(\leq 5)$. Let a
contraction operator $\dot{-}$ be defined as: for any sentence
$\alpha$, $Mod(KB\dot{-}\alpha)=Mod(KB)\cup
Min(Mod(\{\neg\alpha\}\cup IC),\leq_{KB})$. Then $\dot{-}$ satisfies
all the rationality postulates for contraction $(\dot{-}1)$ to
$(\dot{-}8)$.

\begin{proof} Follows from Theorem 1 and Theorem 3.
$\blacksquare$.
\end{proof}
\end{lemma}

\begin{lemma} Let $KB$ be a Horn knowledge base and
$\dot{-}$ a contraction operator that satisfies all the rationality
postulates for contraction $(\dot{-}1)$ to $(\dot{-}8)$. Then, there
exists an order $\leq_{KB}$ among $\mathcal{S}$, that
satisfies$(\leq 1)$ to $(\leq 5)$, and for any sentence $\alpha$,
$Mod(KB\dot{-}\alpha)$ is given as $Mod(KB)\cup
Min(Mod(\{\neg\alpha\}\cup IC),\leq_{KB})$.

\begin{proof} Follows from Theorem 1 and Theorem
3.$\blacksquare$
\end{proof}
\end{lemma}

\begin{theorem} Let $KB$ be contracted by $\alpha$,
and $KB\dot{-}\alpha$ be obtained by the construction discussed
above. Then $\dot{-}$ is a contraction operator iff it satisfies all
the rationality postulates $(\dot{-}1)$ to $(\dot{-}8)$.

\begin{proof} Follows from Lemma 3 and Lemma 4.
$\blacksquare$
\end{proof}
\end{theorem}

\section{Relationship with the coherence approach of $AGM$}
Given Horn knowledge base $KB=<P,Ab,IC,K>$ represents a belief set
$KB^{\bullet}$ that is closed under $Cn_{P}$. We have defined how
$KB$ can be expanded, revised, or contracted. The question now is:
\it does our foundational approach (wrt classical first-order logic)
on $KB$ coincide with coherence approach (wrt our consequence
operator $Cn_{P}$) of $AGM$ on $KB^{\bullet}$?\rm~ There is a
problem in answering this question (similar practical problem
\cite{Ari}) , since our approach, we require $IC$ to be immutable,
and only the current knowledge $K$ is allowed to change. On the
contrary, $AGM$ approach treat every sentence in $KB^{\bullet}$
equally, and can throw out sentences from $Cn_{P}(IC)$. One way to
solve this problem is to assume that sentences in $Cn_{P}(IC)$ are
more entrenched than others. However, one-to-one correspondence can
be established, when $IC$ is empty. The key is our consequence
operator $Cn_{P}$, and in the following, we show that coherence
approach of $AGM$ with this consequence operator, is exactly same as
our foundational approach, when $IC$ is empty.

\subsection{Expansion}
Expansion in $AGM$ (see \cite{Alch})- framework is defined as
$KB\#\alpha=Cn_{P}(KB^{\bullet}\cup\{\alpha\})$, is is easy to see
that this is equivalent to our definition of expansion (when $IC$ is
empty), and is formalized below.

\begin{theorem} Let $KB+\alpha$ be an expansion of $KB$ by $\alpha$ (as defined
in section 3.2). Then $(KB+\alpha)^{\bullet}=KB\#\alpha.$

\begin{proof} By our definition of expansion, $(KB+\alpha)^{\bullet}=Cn_{P}(IC \cup
K\cup\{\alpha\})$, which is clearly the same set as
$Cn_{P}(KB^{\bullet}\cup \{\alpha\})$. $\blacksquare$
\end{proof}
\end{theorem}

\subsection{Revision}

$AGM$ puts forward rationality postulates $(*1)$ to $(*8)$ to be
satisfied by a revision operator on $KB^{\bullet}$. reproduced
below:
\begin{enumerate}
\item[(*1)] (\it Closure) \rm $KB^{\bullet}*\alpha$ is a belief set.
\item[(*2)] (\it Success) $\alpha\in KB^{\bullet}*\alpha$.
\item[(*3)] (Expansion 1) $KB^{\bullet}*\alpha \subseteq KB^{\bullet}\# \alpha$.
\item[(*4)] (Expansion 2) \rm If $\neg \alpha \notin KB^{\bullet},$ then
$KB^{\bullet}\#\alpha\subseteq KB^{\bullet}*\alpha$.
\item[(*5)] (\it Consistency)\rm $KB^{\bullet}*\alpha$ is inconsistent iff
$\vdash \neg \alpha$.
\item[(*6)] (\it Preservation) \rm If $\vdash \alpha \leftrightarrow
\beta$, then $KB^{\bullet}*\alpha=KB^{\bullet}*\beta$.
\item[(*7)] (\it Conjunction 1) $KB^{\bullet}*(\alpha\land \beta)\subseteq
(KB^{\bullet}*\alpha)\#\beta$.
\item[(*8)] (Conjunction 2) \rm If $\neg \beta \notin KB^{\bullet}*\alpha$,
then,$(KB^{\bullet}*\alpha)\#\beta\subseteq
KB^{\bullet}*(\alpha\land \beta)$.
\end{enumerate}

The equivalence between our approach and $AGM$ approach is brought
out by the following two theorems.

\begin{theorem} Let $KB$ a Horn knowledge base with an empty $IC$ and $\dotplus $
be a revision function that satisfies all the rationality postulates
$(\dotplus 1)$ to $(\dotplus 7)$. Let a revision operator $*$ on
$KB^{\bullet}$ be defined as: for any sentence $\alpha$,
$KB^{\bullet}*\alpha=(KB\dotplus \alpha)^{\bullet}$. The revision
operator
*, thus defined satisfies all the $AGM$-postulates for revision
$(*1)$ to $(*8)$.

\begin{proof}~
\begin{enumerate}
\item[(*1)] $KB^{\bullet}*\alpha$ is a belief set.\\
This follows immediately, because $(KB\dotplus \alpha)^{\bullet}$ is
closed wrt $Cn_{P}$.
\item[(*2)] $\alpha\in KB^{\bullet}*\alpha$.\\
This follows from the fact that $\dotplus $ satisfies $(\dotplus
2)$.
\item[(*3)] $KB^{\bullet}*\alpha \subseteq KB^{\bullet}\# \alpha$.\\
\item[(*4)] If $\neg \alpha \notin KB^{\bullet},$ then
$KB^{\bullet}\#\alpha\subseteq KB^{\bullet}*\alpha$.\\
 These two postulates follow
from $(\dotplus 4)$ and theorem 5.
\item[(*5)]$KB^{\bullet}*\alpha$ is inconsistent iff
$\vdash \neg \alpha$.\\
This follows from from $(\dotplus 3)$ and our assumption that $IC$
is empty.
\item[(*6)] If $\vdash \alpha \leftrightarrow
\beta$, then $KB^{\bullet}*\alpha=KB^{\bullet}*\beta$.\\
This corresponds to $(\dotplus 5)$.
\item[(*7)] $KB^{\bullet}*(\alpha\land \beta)\subseteq
(KB^{\bullet}*\alpha)\#\beta$. This follows from $(\dotplus 6)$ and
theorem 5.
\item[(*8)] If $\neg \beta \notin KB^{\bullet}*\alpha$,
then,$(KB^{\bullet}*\alpha)\#\beta\subseteq
KB^{\bullet}*(\alpha\land \beta)$.\\
This follows from $(\dotplus 7)$ and theorem 5. $\blacksquare$
\end{enumerate}
\end{proof}
\end{theorem}

\begin{theorem} Let $KB$ a Horn knowledge base with an empty $IC$ and
* a revision operator that satisfies all the $AGM$-postulates $(*1)$
to $(*8)$. Let a revision function $+$ on $KB$ be defined as: for
any sentence $\alpha$, $(KB\dotplus
\alpha)^{\bullet}=KB^{\bullet}*\alpha$. The revision function $+$,
thus defined, satisfies all the rationality postulates $(\dotplus
1)$ to $(\dotplus 7)$.

\begin{proof}~
\begin{enumerate}
\item[$(\dotplus 1)$] $P,Ab$ and $IC$ do not change.\\
Obvious.
\item[$(\dotplus 2)$] $\alpha$ is accepted in $KB\dotplus \alpha$.\\
Follows from $(^*2)$.
\item[$(\dotplus 3)$] If $\alpha$ is satisfiable and consistent with $IC$,
then $KB\dotplus \alpha$ is consistent.\\
Since we have assumed $IC$ to be empty, this directly corresponds to
$(^*5)$.
\item[$(\dotplus 4)$] If $\neg\alpha$ is not accepted in $KB$, then
$KB\dotplus \alpha\equiv KB+\alpha$.\\ Follows from $(^*3)$ and
$(^*4)$.
\item[$(+5)$] If $KB\equiv KB'$ and $\alpha\equiv \beta$, then
$KB\dotplus \alpha\equiv KB'\dotplus \beta$.\\
Since $KB\equiv KB'$ they represent same belief set, i.e.
$KB^{\bullet}=KB'^{\bullet}$. Now, this postulate follows
immediately from $(^*6)$.
\item[$(\dotplus 6)$] $(KB\dotplus \alpha)+\beta$ implies $KB\dotplus (\alpha\land
\beta)$.\\
Corresponds to $(^*7)$.
\item[$(\dotplus 7)$] If $\neg \beta$ is not accepted in $KB\dotplus \alpha$, then
$KB\dotplus (\alpha\land \beta)$ implies $(KB\dotplus
\alpha)+\beta$.\\
Corresponds to $(^*8)$. $\blacksquare$
\end{enumerate}
\end{proof}
\end{theorem}

\subsection{Contraction}

$AGM$ puts forward rationality postulates $(-1)$ to $(-8)$ to be
satisfied by a contraction operator on closed set $KB^{\bullet}$,
reproduced below:

\begin{enumerate}
\item[$(-1)$] (\it Closure) \rm $KB^{\bullet}-\alpha$ is a belief set.
\item[$(-2)$] (\it Inclusion) \rm $KB^{\bullet}-\alpha\subseteq KB^{\bullet}$.
\item[$(-3)$] (\it Vacuity) \rm If $\alpha \notin KB^{\bullet}$, then
$KB^{\bullet}-\alpha=KB^{\bullet}$.
\item[$(-4)$] (\it Success) \rm If $\nvdash \alpha$, then $\alpha
\notin KB^{\bullet}-\alpha$.
\item[$(-5)$] (\it Preservation) \rm If $\vdash \alpha
\leftrightarrow \beta$, then
$KB^{\bullet}-\alpha=KB^{\bullet}-\beta$.
\item[$(-6)$] (\it Recovery) \rm $KB^{\bullet}\subset (KB^{\bullet}-\alpha)+\alpha$.
\item[$(-7)$] (\it Conjunction 1)\rm $KB^{\bullet}-\alpha\cap
KB^{\bullet}-\beta\subseteq KB^{\bullet}-(\alpha\land\beta)$.
\item[$(-8)$] (\it Conjunction 2) \rm If $\alpha \notin
KB^{\bullet}-(\alpha\land\beta)$, then
$KB^{\bullet}-(\alpha\land\beta)\subseteq KB^{\bullet}-\alpha$.
\end{enumerate}

As in the case of revision, the equivalence is brought out by the
following theorems. Since contraction is constructed in terms of
revision, these theorems are trivial.

\begin{theorem} Let $KB$ be a Horn knowledge base with an empty $IC$ and $\dot{-}$
be a contraction function that satisfies all the rationality
postulates $(\dot{-}1)$ to $(\dot{-}8)$. Let a contraction operator
$-$ on $KB^{\bullet}$ be defined as: for any sentence $\alpha$,
$KB^{\bullet}-\alpha=(KB\dot{-}\alpha)^{\bullet}$. The contraction
operator $-$, thus defined, satisfies all the $AGM$ - postulates for
contraction $(-1)$ to $(-8)$.

\begin{proof} Follows from Theorem 2 and Theorem 6.
$\blacksquare$
\end{proof}

\end{theorem}

\begin{theorem} Let $KB$ be a Horn knowledge base with an empty $IC$ and $-$ be a
contraction operator that satisfies all the $AGM$- postulates $(-1)$
to $(-8)$. Let a contraction function  $\dot{-}$ on $KB$ be defined
as: for any sentence $\alpha$,
$(KB\dot{-}\alpha)^{\bullet}=KB^{\bullet}-\alpha$. The contraction
function $\dot{-}$, thus defined, satisfies all the rationality
postulates $(\dot{-}1)$ to $(\dot{-}8)$.


\begin{proof} Follows from Theorem 3 and Theorem 7.
$\blacksquare$
\end{proof}
\end{theorem}

\section{Realizing Horn knowledge base dynamics using abductive explanations}
In this section, we explore how belief dynamics can be realized in
practice (see \cite{Bess,Arav1,Arav}). Here, we will see how
revision can be implemented based on the construction using models
of revising sentence and an order among them. The notion of
abduction proves to be useful and is explained in the sequel.

Let $\alpha$ be a sentence in $\mathcal{L}$. An \it abductive
explanation \rm for $\alpha$ wrt $KB$ is a set of abductive literals
\footnotemark \footnotetext{An abductive literal is either an
abducible $A$ from $Ab$, or its negation $\neg A$.} $\Delta$ s.t.
$\Delta$ consistent with $IC$ and $\Delta \models_{P}\alpha$ (that
is $\alpha \in Cn_P(\Delta))$. Further $\Delta$ is said to be \it
minimal \rm iff no proper subset of $\Delta$ is an abductive
explanation for $\alpha$.

The basic idea to implement revision of a Horn knowledge base $KB$
by a sentence $\alpha$, is to realize $Mod(\{\alpha\}\cup IC)$ in
terms of abductive explanations for $\alpha$ wrt $KB$. We first
provide a useful lemma.

\begin{definition} Let $KB$ be a Horn knowledge base, $\alpha$ a sentence, and
$\Delta_{1}$ and $\Delta_{2}$ be two minimal abductive explanations
for $\alpha$ wrt $KB$. Then, the \it disjunction \rm of $\Delta_{1}$
and $\Delta_{2}$, written as $\Delta_{1}\lor \Delta_{2}$, is given
as:
$$\Delta_{1}\lor\Delta_{2}=(\Delta_{1}\cap \Delta_{2})\cup
\{\alpha\lor\beta| \alpha \in\Delta_{1}\backslash \Delta_{2}~
\text{and}~\beta\in \Delta_{2}\backslash \Delta_{1}\}.$$ Extending
this to $\Delta^{\bullet}$, a set of minimal abductive explanations
for $\alpha$ wrt $KB$, $\lor\Delta^{\bullet}$ is given by the
disjunction of all elements of $\Delta^{\bullet}$.
\end{definition}

\begin{lemma} Let $KB$ be a Horn knowledge base, $\alpha$ a sentence,and
$\Delta_{1}$ and $\Delta_{2}$ be two minimal abductive explanations
for $\alpha$ wrt $KB$.
Then,$Mod(\Delta_{1}\lor\Delta_{2})=Mod(\Delta_{1})\cup
Mod(\Delta_{2})$.

\begin{proof} First we show that every model of $\Delta_{1}$ is a model of
$\Delta_{1}\lor\Delta_{2}$. Clearly, a model $M$ of $\Delta_{1}$
satisfies all the sentences in $(\Delta_{1}\cap \Delta_{2})$. The
other sentences in $(\Delta_{1}\lor\Delta_{2})$ are of the form
$\alpha \lor\beta$, where $\alpha$ is from $\Delta_{1}$ and $\beta$
is from $\Delta_{2}$. Since $M$ is a model of $\Delta_{1}$, $\alpha$
is true in $M$, and hence all such sentences are satisfied by $M$.
Hence $M$ is a model of $\Delta_{1}\lor\Delta_{2}$ too. Similarly,
it can be shown that every model of $\Delta_{2}$ is a model of
$\Delta_{1}\lor\Delta_{2}$ too.

Now, it remains to be shown that every model $M$ of
$\Delta_{1}\lor\Delta_{2}$ is either a model of $\Delta_{1}$ or a
model of $\Delta_{2}$. We will now show that if $M$ is not a model
of $\Delta_{2}$, then it must be a model of $\Delta_{1}$. Since $M$
satisfies all the sentences in $(\Delta_{1}\cap\Delta_{2})$, we need
only to show that $M$ also satisfies all the sentences in
$\Delta_{1}\backslash\Delta_{2}$. For every element $\alpha \in
\Delta_{1}\backslash\Delta_{2}$: there exists a subset of
$(\Delta_{1}\lor \Delta_{2})$, $\{\alpha\lor \beta|\beta \in
\Delta_{2}\backslash\Delta_{2}\}$. $M$ satisfies all the sentences
in this subset. Suppose $M$ does not satisfy $\alpha$, then it must
satisfy all $\beta\in \Delta_{1}\backslash\Delta_{2}$. This implies
that $M$ is a model of $\Delta_{2}$, which is a contradictory to our
assumption. Hence $M$ must satisfy $\alpha$, and thus a model
$\Delta_{1}$. Similarly, it can be shown that $M$ must be a model of
$\Delta_{2}$ if it is not a model of $\Delta_{1}$. $\blacksquare$
\end{proof}
\end{lemma}

As one would expect, all the models of revising sentence $\alpha$
can be realized in terms abductive explanations for $\alpha$, and
the relationship is precisely stated below.

\begin{lemma} Let $KB$ be a Horn knowledge base, $\alpha$ a sentence, and
$\Delta^{\bullet}$ the set of all minimal abductive explanations for
$\alpha$ wrt $KB$. Then $Mod(\{\alpha\}\cup
IC)=Mod(\lor\Delta^{\bullet})$.

\begin{proof} It can be easily verified that every model $M$ of a minimal
abductive explanation is also a model of $\alpha$. Since every
minimal abductive explanation satisfies $IC$, $M$ is a model of
$\alpha\cup IC$. It remains to be shown that every model $M$ of
$\{\alpha\}\cup IC$ is a model of  one of the minimal abductive
explanations for $\alpha$ wrt $KB$. This can be verified by
observing that a minimal abductive explanation for $\alpha$ wrt $KB$
can be obtained from $M$. $\blacksquare$
\end{proof}
\end{lemma}

Thus, we have a way to generate all the models of $\{\alpha\}\cup
IC$, and we just need to select a subset of this based on an order
that satisfies $(\leq 1)$ to $(\leq 5)$. Suppose we have such an
order that satisfies all the required postulates, then this order
can be mapped to a particular set of abductive explanations for
$\alpha$ wrt $KB$. This is stated precisely in the following
theorem. An important implication of this theorem is that there is
no need to compute all the abductive explanations for $\alpha$ wrt
$KB$. However, it does not say which abductive explanations need to
be computed.

\begin{theorem} Let $KB$ be a Horn knowledge base, and $\leq_{KB}$ be an order
among abductive interpretations in $\mathcal{S}$ that satisfies all
the rationality axioms $(\leq 1)$ to $(\leq 5)$. Then, for every
sentence $\alpha$, there exists $\Delta^{\bullet}$ a set of minimal
abductive explanations for $\alpha$ wrt $KB$, s.t.
$Min(Mod(\{\alpha\}\cup IC),\leq_{KB})$ is a subset of $Mod(\lor
\Delta^{\bullet})$, and this does not hold for any proper subset of
$\Delta^{\bullet}$.

\begin{proof} From Lemma 6. and Lemma 5., it is clear that
$Mod(\{\alpha\}\cup IC)$ is the union of all the models of all
minimal abductive explanations of $\alpha$ wrt $KB$. $Min$ selects a
subset of this, and the theorem follows immediately. $\blacksquare$.
\end{proof}
\end{theorem}

The above theorem 10. is still not very useful in realizing
revision. We need to have an order among all the interpretations
that satisfies all the required axioms, and need to compute all the
abductive explanations for $\alpha$ wrt $KB$. The need to compute
all abductive explanations arises from the fact that the converse of
the above theorem does not hold in general. This scheme requires an
universal order $\leq$, in the sense that same order can be used for
any Horn knowledge base. Otherwise, it would be necessary to specify
the new order to be used for further modifying $(KB\dotplus\alpha)$.
However, even if the order can be worked out, it is not desirable to
demand all abductive explanations of $\alpha$ wrt $KB$ be computed.
So, it is desirable to work out, when the converse of the above
theorem is true. The following theorem says that, suppose $\alpha$
is rejected in $KB$, then revision of $KB$ by $\alpha$ can be worked
out in terms of some abductive explanations for $\alpha$ wrt $KB$.

\begin{theorem} Let $KB$ be a Horn knowledge base, and a revision function
$\dotplus$ be defined as: for any sentence $\alpha$ that is rejected
in $KB$, $Mod(KB\dotplus\alpha)$ is a non-empty subset of
$Mod(\lor\Delta^{\bullet})$, where $\Delta^{\bullet}$ is a set of
all minimal abductive explanations for $\alpha$ wrt $KB$. Then,
there exists an order $\leq_{KB}$ among abductive interpretations in
$\mathcal{S}$, s.t. $\leq_{KB}$ satisfies all the rationality axioms
$(\leq 1)$ to $(\leq 5)$ and $Mod(KB\dotplus\alpha)=
Min(Mod(\{\alpha\}\cup IC),\leq_{KB})$.

\begin{proof} It is easy to define a pre-order s.t. every model of
$Mod(KB\dotplus\alpha)$ is strictly minimal than all other
interpretations. It is easy to verify that such a pre-order
satisfies $(\leq 1)$ to $(\leq 5)$. In particular, since $\alpha$ is
rejected in $KB$, $(\leq 3)$ faithfulness is satisfied, and since
non-empty subset of $Mod(\lor\Delta^{\bullet})$ is selected, $(\leq
4)$ is also satisfied. $\blacksquare$
\end{proof}
\end{theorem}

An important corollary of this theorem is that, revision of $KB$ by
$\alpha$ can be realized just by computing \it one \rm abductive
explanation of $\alpha$ wrt $KB$, and is stated below.

\begin{corollary} Let $KB$ be a Horn knowledge base, and a revision function
$\dotplus$ be defined as: for any sentence $\alpha$ that is rejected
in $KB$, $Mod(KB\dotplus\alpha)$ is a non-empty subset of
$Mod(\Delta)$, where $\Delta$ is an abductive explanations for
$\alpha$ wrt $KB$. Then, there exists an order $\leq_{KB}$ among
abductive interpretations in $\mathcal{S}$, s.t. $\leq_{KB}$
satisfies all the rationality axioms $(\leq 1)$ to $(\leq 5)$ and
$Mod(KB\dotplus\alpha)= Min(Mod(\{\alpha\}\cup IC),\leq_{KB})$.
$\blacksquare$
\end{corollary}

The precondition that $\alpha$ is rejected in $KB$ is not a serious
limitation in various applications such as database updates and
diagnosis, where close world assumption is employed to infer
negative information. For example, in diagnosis it is generally
assumed that all components are functioning normally, unless
otherwise there is specific information against it. Hence, a Horn
knowledge base in diagnosis either accepts or rejects normality of a
component, and there is no "don't know" third state. In other words,
in these applications the Horn knowledge base is assumed to be
complete. Hence, when such a complete Horn knowledge base is revised
by $\alpha$, either $\alpha$ is already accepted in $KB$ or rejected
in $KB$, and so the above scheme works fine.

\section{Related Works}

We begin by recalling previous work on view deletion. Chandrabose
\cite{Arav1,Arav}, defines a contraction operator in view deletion
with respect to a set of formulae or sentences using Hansson's
\cite{Hans2} belief change. Similar to our \cite{Del} approach, he
focused on set of formulae or sentences in Horn knowledge base
revision for view update wrt. insertion and deletion and formulae
are considered at the same level. Chandrabose proposed different
ways to change Horn knowledge base via only database deletion,
devising particular postulate which is shown to be necessary and
sufficient for such an update process.

Our Horn knowledge base consists of two parts, immutable part and
updatable part , but focus is on principle of minimal change. There
are more related works on that topic. Eiter \cite{Eit} is focusing
on revision from different perspective - prime implication.
Segerberg \cite{Seg} defined new modeling for belief revision in
terms of irrevocability on prioritized revision. Hansson
\cite{Hans2} constructed five types of non-prioritized belief
revision. Makinson \cite{Mak} developed dialogue form of revision
AGM. Papini\cite{Pap} defined a new version of Horn knowledge base
revision.

We are bridging gap between philosophical work, paying little
attention to computational aspects of database work \cite{Min,Sie}.
In such a case, Hansson's\cite{Hans2} kernel change is related with
abductive method. Aliseda's \cite{Alis} book on abductive reasoning
is one of the motivation keys. Christiansen's \cite{Chris1,Chris}
work on dynamics of abductive logic grammars exactly fits our
minimal change (insertion and deletion).

In general, our abduction theory is related to Horn knowledge base
dynamics (see how abduction theory is related with other
applications, respectively, reasoning \cite{Bar,Sad,Sak1},
update\cite{Sak3,Sak4}, equivalence{\cite{Ino,Sak1,Sak2} and problem
solving{\cite{Ino1,Lobo1}). More similar to our work is paper
presented by Bessant et al. \cite{Bess}, local search-based
heuristic technique that empirically proves to be often viable, even
in the context of very large propositional applications. Laurent et
al.\cite{Lau} parented updating deductive databases in which every
insertion or deletion of a fact can be performed in a deterministic
way.

Furthermore, and at a first sight more related to our work, some
work has been done on "core-retainment" (same as our immutable part)
in the model of language splitting introduced by Parikh \cite{Par}.
More recently, Doukari \cite{Dou}, \"{O}z\c{c}ep \cite{Ozc} and Wu,
et al. \cite{Wu} applied similar ideas for dealing with knowledge
base dynamics. These works represent motivation keys for our future
work. Second, we are dealing with how to change minimally in the
theory of "principle of minimal change", but current focuss is on
finding second best abductive explanation \cite{Lib} and 2-valued
minimal hypothesis for each normal program \cite{Pin}. Finally, when
we presented Horn knowledge base change in abduction framework, we
did not talk about compilability and complexity (see the works of
Liberatore \cite{Lib1} and Zanuttini \cite{Zan}).
\section{Conclusion}
The main contribution of this work lies in showing how abductive
framework deals with Horn knowledge base dynamics via belief change
operation. We consider the immutable part as defining a new logical
system. By a logical system, we mean that it defines its own
consequence relation and closure operator. We presented that
relationship of the coherence approach of $AGM$ with this
consequence operator is exactly same as our foundational approach,
when $IC$ is empty.

We believe that Horn knowledge base dynamics can also be applied to
other applications such as view maintenance, diagnosis, and we plan
to explore it in further works \cite{Caro}. Still, a lot of
developments are possible, for improving existing operators or for
defining new classes of change operators. As immediate extension,
question raises: is there any \emph{real life application for AGM in
25 year theory?} \cite{Ferme}. The revision and update are more
challenging in Horn knowledge base dynamic, so we can extend the
theory to combine results similar to Konieczny's \cite{Kon} and
Nayak's \cite{Nayak2}.

\section*{Acknowledgement}

The author acknowledges the support of RWTH Aachen, where he is
visiting scholar with an Erasmus Mundus External Cooperation Window
India4EU by the European Commission when the paper was written. I
would like to thanks Chandrabose Aravindan and Gerhard Lakemeyer
both my Indian and Germany PhD supervisor, give encourage to write
the paper.


\end{document}